\documentclass[journal,twoside,web]{ieeecolor}

\usepackage{generic}
\usepackage{cite}
\usepackage{amsmath,amssymb,amsfonts}
\usepackage{algorithmic}
\usepackage{graphicx}
\usepackage{textcomp}
\usepackage{multirow}

\usepackage{comment}
\usepackage{balance}
\usepackage{graphicx}
\usepackage{subfigure}
\usepackage{epsfig} 
\usepackage{cite}
\usepackage{lcsys}

\usepackage[colorlinks = true,
            linkcolor = blue,
            urlcolor  = blue,
            citecolor = blue,
            anchorcolor = blue]{hyperref}

\usepackage{etoolbox}
\makeatletter
\@ifundefined{color@begingroup}%
  {\let\color@begingroup\relax
   \let\color@endgroup\relax}{}%
\def\fix@ieeecolor@hbox#1{%
  \hbox{\color@begingroup#1\color@endgroup}}
\patchcmd\@makecaption{\hbox}{\fix@ieeecolor@hbox}{}{\FAILED}
\patchcmd\@makecaption{\hbox}{\fix@ieeecolor@hbox}{}{\FAILED}

\newtheorem{theorem}{Theorem}
\newtheorem{lemma}{Lemma}

\newtheorem{corollary}{Corollary}
\newtheorem{definition}{Definition}
\newtheorem{remark}{Remark}
\newtheorem{assumption}{Assumption}

\renewcommand{\baselinestretch}{.985}

\begin{document}

\def\BibTeX{{\rm B\kern-.05em{\sc i\kern-.025em b}\kern-.08em
    T\kern-.1667em\lower.7ex\hbox{E}\kern-.125emX}}
\markboth{\journalname, VOL. XX, NO. XX, XXXX 2017}
{Author \MakeLowercase{\textit{et al.}}: Preparation of Papers for IEEE Control Systems Letters (August 2022)}

\title{\huge{Fixed-Time Voltage Regulation in Distribution Networks with Impedance Awareness}}


\author{
Nilanjan~Roy~Chowdhury$^{\star}$
and Venkatesh~Sarangan$^{\star}$,
\thanks{$^{\star}$~N. Roy Chowdhury and V. Sarangan are with the Future of Energy Research Division, Tata Consultancy Services Research, Pune 411057, India; \texttt{\{nilanjanroy.c,venkatesh.sarangan\}@tcs.com}}
\vspace{-5mm}
}

\maketitle
\thispagestyle{empty}

\begin{abstract}
This letter introduces an optimization-based fixed time control algorithm for solving the voltage regulation problem of a radial and balanced power distribution network. The proposed algorithm requires no prior knowledge of the exact network impedance (resistance and reactance), yet guarantees voltage convergence to the predefined safe limit within a fixed time-window. We embrace results from \textit{Fixed-time stability} (FxTs) and \textit{Control Lyapunov function} (CLF) to analyze the stability and robustness of the underlying algorithm and then synthesize it leveraging the \textit{Quadratic Programming} (QP) approach. We first analytically provide sufficient conditions on the control gains and the design parameters, that can ensure voltage regulation in fixed-time. Thereafter we transform the existing voltage regulation problem into an equivalent QP-based optimization framework and translate the aforesaid conditions to find feasible solutions for voltage stability. We empirically verify the efficacy of our proposed algorithm on the IEEE-33 bus distribution network and compare its performance against other existing robust control methods.
\end{abstract}


\vspace{-0.2 cm}

\section{Introduction}
Voltage regulation problem in power distribution networks has drawn a substantial research attention in recent times, where the central objective is to minimize the \textit{voltage recovery time}\footnote{\textit{Voltage recovery time} refers to the time duration that required to bring back the network voltage within a pre-defined safe limits (i.e., $\pm\,5\%$ around the nominal voltage $\mathbf{v}_{\text{nom}}$) following the occurrence of a sudden and large external disturbance.}~\cite{shi2022stability}. These problems are inherently non-convex, and their complexity can further be exacerbated if the underlying network impedance values are uncertain~\cite{yeh2024online,chowdhury2025provably}. Despite this challenge, over the years, numerous feedback-based voltage regulation algorithms have been surfaced as a solution to this problem, see~\cite{shi2022stability,farivar2015local,qu2019optimal,colot2024optimal,chowdhury2025new,feng2023bridging} and others. However, most of these solutions yield asymptotic (or rather exponential) \textit{voltage recovery time}. Therefore, this letter focuses on developing a sophisticated \textit{feedback-based} robust control method that can solve the voltage regulation problem of a distribution network considering uncertainty in the network impedance values and furnish an improved voltage recovery time.

\smallskip

\noindent \textit{Prior studies}:~Traditional feedback-based voltage regulation algorithms mostly rely on either \textit{local} or \textit{distributed} control architectures. The \textit{local} control frameworks consider their own voltage measurements to compute the reactive power set-points~\cite{farivar2015local}, while in the \textit{distributed} formulations, the controllable nodes share information with their neighbours and steer the network voltage towards the optimal solution~\cite{qu2019optimal}. However, these results provide asymptotic voltage recovery time. Furthermore,~\cite{li2014real} infers that \textit{local} controllers fail to maintain a feasible voltage profile under certain circumstances and may experience stability issues for large-scale networks. Recently, \textit{centralized} robust control methods have also been envisioned in~\cite{colot2024optimal,chowdhury2025new} to deal with voltage regulation problems in real time with impedance uncertainty. These works leverage safety-critical control frameworks to ensure voltage converges to the safe limits, and employ a convex QP framework at each iteration to compute the control set-points. Unfortunately, the analysis presented in these studies only guarantees \textit{exponentially stable} solutions, implying that the convergence remains asymptotic rather than strictly bounded in time.

Recently, data-driven and machine learning approaches have also emerged to address network uncertainty in voltage regulation problems. \textit{Data-driven techniques} typically linearize non-linear branch flow equations to estimate network impedance offline using smart-meter data~\cite{picallo2022adaptive}. However, they primarily prioritize asymptotic convergence, and their efficacy critically depends on the quality and quantity of available data. Similarly, \textit{learning-based algorithms} leverage historical network data and synthesize \textit{stable} voltage regulators without knowing exhaustive network models, where the stability verifications are performed employing Lyapunov function theory~\cite{shi2022stability,feng2023bridging}. However, these verifications hold over an infinite time horizon, and require ample amounts of historical training data to achieve satisfactory performances. 

To enhance the efficacy of voltage regulation algorithms and to find an ameliorated voltage recovery time, fast/~accelerated voltage regulation methods have also been advocated in~\cite{kekatos2015fast,zhou2019accelerated}. These methods model (controllable) reactive power updates of a distribution network using traditional proximal gradient algorithms, and propose new scaling factors to improve their convergence rates. Albeit, these modifications can only attain exponential convergence and an improved result beyond that remains open.

\smallskip

\noindent \textit{Summary of Contributions}:~This letter addresses the aforesaid open problem by introducing a feedback-based \textit{fixed-time} control algorithm to solve the real-time voltage regulation problem of a \textit{balanced} distribution network under impedance uncertainties. Moving beyond existing methods that only offer infinite-settling-time guarantees, our approach ensures a pre-calculable, uniformly bounded recovery time. To ensure model accuracy of the network, we eschew traditional lossless formulations in~\cite{farivar2015local,shi2022stability,yeh2024online,qu2019optimal,feng2023bridging} that omit network losses and instead employ the \textit{lossy linearization} method from \cite{schweitzer2019lossy}, which analytically approximates network losses and offers significant benefits. The proposed algorithm is built on a novel framework that integrates fixed-time stability (FxTS), control Lyapunov functions (CLF), and quadratic programming (QP). Within this structure, FxTS and CLF are leveraged to analytically determine the control gain and design parameter bounds necessary for stability and robustness, while the QP framework enhances computational tractability and provides feasible control set-points. In essence, the proposed algorithm accounts for impedance uncertainties, adheres to reactive power limits, optimizes transient performance, and guarantees closed-loop stability while improving settling time.

\smallskip

\noindent \textit{Notation:} We use standard notation: $\mathbf{1}_n$ is the all-ones column vector, and $\mathbf{I}_n$ is the identity matrix (subscripts omitted if unambiguous). $\nabla h(\cdot)$ is the gradient of $h(\cdot)$, and $\|\cdot\|$ refers its Euclidean norm. For matrix $\mathbf{A}$, $\mathbf{A}^{-1}$ and $\mathbf{A}^{-\top}$ denote its inverse and inverse transpose. $\mathrm{diam}(\mathcal{X})$ is the diameter of set $\mathcal{X}$, $P\succ 0$ indicates a symmetric, positive-definite matrix.

\section{Foundational Frameworks}
Here, in Section~\ref{distflow}, we first briefly review the linearized \textit{lossy distribution flow} model (see~\cite{schweitzer2019lossy} for details), and then adopting this model, in section~\ref{revise}, we formulate the voltage regulation problem with \textit{fixed} voltage recovery time.  

\subsection{Linear approximation of branch flow model} \label{distflow}
Power distribution networks are typically radial in nature, and thus represented by a tree graph $\mathcal{T}=\left( \mathcal{N},\,\mathcal{E}\right)$. Here $\mathcal{N}\,=\,\left\lbrace 0,1,\cdots,n\right\rbrace$ is the node set, in which the root node $\left\lbrace 0\right\rbrace $ refers the \textit{substation node}, and its voltage $\mathbf{v}_0$ is assumed to be known and fixed. Likewise, the feeder lines of the underlying network are represented by the edge set $\mathcal{E}\,=\,\left\lbrace \left(i,j\right)\,:\,i,\,j\in\mathcal{N}\right\rbrace $. The topology of $\mathcal{T}$ is captured by an \textit{edge-to-node} incidence matrix $\bar{\mathbf{A}}\,:=\,\begin{bmatrix}\mathrm{a} & \mathbf{A}\end{bmatrix}$, where removing the root-node column $\mathrm{a}\in\left\lbrace 0,\,1\right\rbrace^{n}$ yields the non-singular \textit{reduced} incidence matrix $\mathbf{A}\in\left\lbrace 0,\,\pm\,1\right\rbrace^{n\times\,n}$. Let $\mathbf{v}(t),\,\mathbf{p}(t),\,\mathbf{q}(t)\in\mathbb{R}^n$ are respectively denoted as the voltage magnitude, active and reactive power injection values of the \textit{non-substation} nodes. Then, following the lossy distribution flow model in~\cite{schweitzer2019lossy} and employing the methodology in~\cite{chowdhury2025new}, we linearly approximate $\mathbf{v}(t)$ as
\begin{multline} \label{ldf}
\left(\mathbf{v}(t)-\mathbf{v}_0\mathbf{1}_{n}\right)\,=\,\underbrace{\mathbf{A}^{-1}\left( \mathbf{I}\,+\,\Phi\Psi\right)^{-1} \mathcal{D}_{\mathbf{x}}\mathbf{A}^{-\top}}_{\mathfrak{B}}\,\mathbf{q}(t) \\
+\,\underbrace{\mathbf{A}^{-1}\left( \mathbf{I}\,+\,\Phi\Psi\right)^{-1}\mathcal{D}_{\mathbf{r}}\mathbf{A}^{-\top}}_{\mathfrak{D}}\,\mathbf{p}(t),
\end{multline}
where, 
\begin{subequations} \label{dan100}
\begin{align}
\Phi\,&=\,\mathcal{D}^2_{\mathbf{r}}\,+\,\mathcal{D}^2_{\mathbf{x}}\,+\,2\mathcal{D}_{\mathbf{r}}\mathbf{A}^{-\top}\mathcal{D}_\mathbf{r}\,+\,2\mathcal{D}_\mathbf{x}\mathbf{A}^{-\top}\mathcal{D}_\mathbf{x},\\
\Psi\,&=\,\left(\mathcal{D}^2_{\mathbf{r}}\,+\,\mathcal{D}^2_{\mathbf{x}}\right)^{-1}\left(\mathbf{I}-2\mathcal{D}_{\sigma}\right).
\end{align}
\end{subequations}
Here the diagonal matrices $\mathcal{D}_{\mathbf{r}},\mathcal{D}_{\mathbf{x}}$ and $\mathcal{D}_{\sigma}$ embed the network resistances $\mathbf{r}$, reactances $\mathbf{x}$ and the \textit{voltage sensitivity parameter} $\sigma$ on their main diagonals. The sensitivity parameter $\sigma_{ij}$ approximates losses of the $ij$th feeder line by representing its end point voltages in terms of their squared magnitude values~\cite{schweitzer2019lossy}. Subsequently, Remarks~\ref{threm} and \ref{ben} outline the physical significance of $\sigma$ and the advantages of the proposed model \eqref{ldf} over the standard \textit{lossless flow model}, respectively.

\begin{remark} \label{threm}
Given a distribution network, the parameter $\sigma_{ij}\in\,\left(0,\,1\right)$ may vary across different feeder lines. Notably, one can derive the standard lossless model used in \cite{farivar2015local, yeh2024online, shi2022stability, feng2023bridging} by setting $\sigma_{ij}\,=\,0.5$ for all feeder lines in~\eqref{ldf}. Thus,~\eqref{ldf} serves as a generalized linear approximation framework that encapsulates the lossless model as a special case. Consequently, all subsequent results can directly be applied to lossless networks without any modifications.
\end{remark}

\smallskip

\begin{remark} \label{ben}
The lossy flow model \eqref{ldf} offers several advantages over its lossless counterpart. Notably, the lossless model overestimates voltage magnitudes and underestimates power flows in radial networks with positive line impedances \cite{markovic2022parameterized}. Furthermore, its accuracy degrades outside a tightly bounded, lightly loaded operating region. Conversely, accounting for network losses consistently yields high accuracy in state estimation across varying loading levels.~\cite{schweitzer2019lossy}.
\end{remark}

\subsection{Problem setup} \label{revise}
This section formulates the voltage regulation task as a fixed-time optimal control problem stated below:
\begin{subequations} \label{rc100}
\begin{align}
\min_{u(t)} \quad & \int_{t_0}^{t_f}\,\biggl(\sum_{i\in\mathcal{N}}\varrho^i_{v}\,\phi_i(\mathbf{v}_i)\,+\,\sum_{i\in\mathcal{C}}\varrho^i_{u}\,\psi_i(u_i)\biggr)\,\mathrm{dt} \label{rc200}\\
\textrm{s.\,t.} \quad & \mathbf{v}(t)\,=\,\mathfrak{B}\,\mathbf{q}(t)\,+\,\underbrace{\mathfrak{D}\,\mathbf{p}(t)\,+\,\mathbf{v}_0\mathbf{1}_{n}}_{\mathbf{v}_{\mathrm{ext}}(t)}\label{rc300}\\
  & \dot{\mathbf{q}}(t)\,=\,u(t) \label{rc400}\\
  & u(t)\,\in\,\mathcal{X}_{u} \quad \forall\,t\,\in\,\left[ t_0,\,t_f\right] \label{rc500}\\
  & \mathbf{v}(t)\in \mathcal{X}_v \quad\,\forall\,t\in\,\left[ t_{\mathrm{FX}},\,t_f\right]\subseteq\,\left[ t_0,\,t_f\right] \label{rc600}\\
  & t_{\mathrm{FX}}~\text{is pre-defined.} \label{rc700}
\end{align}
\end{subequations}   
Here given a fixed time $t_{\mathrm{FX}}\,\in\,\left( t_0,\,t_f\right)$, one aims to jointly minimize the voltage violation cost $\phi_i(\cdot)$ and the control (reactive power) cost $\varphi_i(\cdot)$, while the coefficients $\varrho^i_{\phi},\,\varrho^i_{\varphi}$ are used to balance these costs. We presume that the rate of change of reactive power injection $\mathbf{q}(t)$ (referred as $u(t)$) in~\eqref{rc400} is the decision variable, and this injection rate is bounded by the compact set $\mathcal{X}_u\,:=\,\left[ -u_{\max},\,u_{\max}\right] $ given in~\eqref{rc500}. Based on~\eqref{ldf}, the linearized voltage-reactive power dynamics is formulated as a continuous time dynamical systems in \eqref{rc300} to capture fast-acting reactive power control loops. This continuous-time formulation contrasts with the conventional discrete-time models adopted in \cite{li2014real, yeh2024online}. In~\eqref{rc300}, the drift matrix $\mathfrak{B}\in\mathbb{R}^{n\times n}$ is constructed from the network resistance and reactance, while $\mathbf{v}_{\mathrm{ext}}(t)\in\mathbb{R}^n$ accounts for \textit{uncontrollable} active power injections of the \textit{non-substation nodes} and the \textit{substation node} voltage. Throughout this work, $\mathbf{v}_{\mathrm{ext}}(t)$ is assumed to be unknown in real time, which precludes the use of feedforward optimization techniques for computing the control set-points. We aim to solve \eqref{rc100} by finding an optimal $u(t)\in\mathcal{X}_u$, that reinstates network voltages into the safe limit $\mathcal{X}_v$\footnote{Here $\mathcal{X}_v\,=\,\left[ \mathbf{v}_{\min},\,\mathbf{v}_{\max}\right]\,=\,\left[ 0.95\mathbf{v}_{\mathrm{nom}},\,1.05\mathbf{v}_{\mathrm{nom}}\right]$, in which $\mathbf{v}_{\min}, \mathbf{v}_{\max}$ and $\mathbf{v}_{\mathrm{nom}}$, respectively refer as the minimum, maximum voltage tolerance limit, and the nominal network voltage.} at a predefined fixed time $t_{\mathrm{FX}}$ after the occurrence of an external disturbance.

\smallskip

\begin{remark}
Note that, in~\eqref{rc100}, we presume that the controllers are placed at each node of a given distribution network, i.e., all nodes are controllable. However, we contend that our aforesaid formulation and the subsequent findings can also be applied to the distribution networks, in which only a subset of nodes are controllable. To explicate this statement, we partition the node set $\bar{\mathcal{N}}\,:=\,\mathcal{N}\setminus\,\left\lbrace 0\right\rbrace $ of a distribution network into two subsets, i.e., controllable nodes, denoted by $\mathcal{C}$, and uncontrollable nodes, defined as $\mathcal{L}\,=\,\bar{\mathcal{N}}\setminus\mathcal{C}$. Then rearranging the controllable and uncontrollable nodes, we write~\eqref{ldf} as
\begin{align}
\mathbf{v}(t)\,&=\,\begin{bmatrix}
\mathfrak{B}_{\mathcal{C}} & \mathfrak{G}_1\\
\mathfrak{G}_2 & \mathfrak{B}_{\mathcal{L}}
\end{bmatrix}\,\begin{bmatrix}
\mathbf{q}_{\mathcal{C}}(t)\\ \mathbf{q}_{\mathcal{L}}(t)
\end{bmatrix}\,+\,\mathfrak{D}\,\mathbf{p}(t)\,+\,\mathbf{v}_0\mathbf{1}_{n},\nonumber \\
&=\,\underbrace{\begin{bmatrix}
\mathfrak{B}_{\mathcal{C}} \\ \mathfrak{G}_2\end{bmatrix}}_{\mathrm{F}}\,\mathbf{q}_{\mathcal{C}}(t)\,+\,\underbrace{\begin{bmatrix}
\mathfrak{G}_1 \\ \mathfrak{B}_{\mathcal{L}}
\end{bmatrix}\,\mathbf{q}_{\mathcal{L}}(t)\,+\,\mathfrak{D}\,\mathbf{p}(t)\,+\,\mathbf{v}_0\mathbf{1}_{n}}_{\mathbf{v}_{\mathrm{ext}}(t)}, \label{ds123}
\end{align}
where the matrices $\mathfrak{B}_{\mathcal{C}},\,\mathfrak{B}_{\mathcal{L}}, \mathfrak{G}_1$ and $\mathfrak{G}_2$ are with appropriate dimensions. It is observed that the structure of~\eqref{ds123} adheres the \textit{same} characteristics of the continuous time dynamical systems given in~\eqref{rc300}. Thus, the proposed fixed-time algorithm and its corresponding results can easily be invoked to~\eqref{ds123} without any exhaustive modifications.    
\end{remark}

\section{Characteristics and Behaviour of Fixed-time Voltage Regulators}
This section introduces a fixed time voltage regulation algorithm for solving~\eqref{rc100}. Specifically, to solve~\eqref{rc100}, one requires the \textit{exact} values of the $\mathfrak{B}$ matrix, which depends on network resistance and reactance. However, factors like system aging, degradation, and thermal loading often make this information unavailable or outdated. Our proposed algorithm circumvents this challenge by solving~\eqref{rc100} without knowing $\mathfrak{B}$ precisely. To analyse this algorithm, we first derive the voltage-reactive power dynamics from the system model~\eqref{ldf}, and then transform it to an equivalent stabilization problem. We define $e(t)\,:=\,\mathbf{v}(t)-\mathbf{v}_{\text{nom}}$ as the \textit{voltage violation vector} and represent~\eqref{rc300} as a continuous-time dynamical system: 
\begin{align}
\frac{\mathrm{d}}{\mathrm{dt}}e(t)\,=\,\hat{\mathfrak{B}}\,u(t)\,+\,\underbrace{\left( \mathfrak{B}-\hat{\mathfrak{B}}\right)\,u(t)}_{\chi(u(t))}\,+\,\bigl( w(t)\,+\,d(t)\bigr). \label{sf100}
\end{align}
Here $\chi(\cdot)$ is modelled as a vanishing perturbation, in which $\hat{\mathfrak{B}}$ denotes the estimate of $\mathfrak{B}$. Conversely, we refer $w(t),\,d(t)\in\,\mathbb{R}^n$ as the external disturbances, in which $w(t):=\,\frac{\mathrm{d}}{\mathrm{dt}}\,\mathbf{v}_{\mathrm{ext}}(t)$ is the~\textit{persistent disturbance}, and $d(t)$ is the \textit{momentary disturbance}. Hereinafter, we consider $d(t)$ in \eqref{sf100}, acts as an \textit{impulsive transient}, exhibiting sharp rise and decay times. We make the following assumptions for~\eqref{sf100}:

\smallskip

\begin{assumption}\label{bound1}
Under all operating conditions, the \textit{persistent disturbance} $w(t)$ is bounded, i.e., $||w(t)||\,\leq\,\gamma$, where $\gamma=\left( \mathbf{v}_{\max}-\mathbf{v}_{\min}\right)/2$. 
\end{assumption}

\smallskip

Assumption~\ref{bound1} implies that the active bus power injection values do not vary abruptly over time, which seems to be reasonable, as the load profile of distribution networks are sufficiently smooth. Moreover, the definition of $\gamma$ infers that whenever $\mathbf{v}(t)$ stays within $\mathcal{X}_v$, any $w(t)$ satisfying Assumption~\ref{bound1}, does not push it outside $\mathcal{X}_v$.

\smallskip

\begin{assumption} \label{mda}
There exist a (known) $d_{\max}\,>\,0$, such that for all $t\,\geq\,0$, it holds $||d(t)||\,\leq\,d_{\max}$.
\end{assumption}

\smallskip

The above assumption infers that the magnitude and duration of the \textit{momentary disturbance} $d(t)$ cannot grow infinitely and are constrained by the physical laws, architecture, and protection systems of the grid, see~\cite{zhang2024online}.

\smallskip

\begin{assumption}\label{bound2}
There exist an $\delta\,>\,0$ such that $\big|\big|\mathfrak{B}\,-\,\hat{\mathfrak{B}}\big|\big|\,\leq\,\delta$.
\end{assumption}

\smallskip

This assumption is plausible as most existing algorithms yield bounded impedance estimates; see~\cite{yeh2024online,picallo2022adaptive}.


\subsection{Stability Studies of Fixed-Time Regulation Paradigm}\label{sfxt}
This section investigates the stability properties of the fixed-time voltage regulation algorithm. This investigation relies on the notion of FxTS-CLF, which is defined below:

\smallskip

\begin{definition}[\textit{FxTS-CLF-$\mathcal{S}_G$}] \label{fxt}
Consider the following nonlinear control affine system
\begin{align}
\frac{\mathrm{d}}{\mathrm{dt}}z(t)\,=\,f(z)\,+\,g(z)\,u(t), \label{abc}
\end{align}
where $z\in\mathcal{Z}\subseteq\,\mathbb{R}^n$, $u\in\mathcal{U}\subseteq\,\mathbb{R}^m$ and $f:\mathcal{Z}\rightarrow\,\mathbb{R}^n$ (resp. $g:\mathcal{Z}\rightarrow\,\mathbb{R}^{n\times\,m}$) describes the drift (resp. the control) vector field. Let $\mathcal{Z}$ be a compact (closed and bounded) set containing the goal set $\mathcal{S}_G$ in its interior. Then a continuously differentiable function $h:\mathcal{Z}\rightarrow\,\mathbb{R}_{+}$ is said to be a FxTS-CLF with respect to the set $\mathcal{S}_G$, if it is positive definite and satisfies the following: 
\begin{subequations} 
\begin{align*}
h(z)\,&\leq\,0 \quad \forall z\in\,\mathcal{S}_G \\
h(z)\,&>\,0 \quad \forall z\in\,\mathcal{Z}\setminus\,\mathcal{S}_G \\
\inf_{u\in\mathcal{U}}\nabla\,h(z)^{\top}\left(f(z)+g(z)u\right)&\leq-\alpha_1\,h(z)^{\mu_1}-\alpha_2\,h(z)^{\mu_2}, 
\end{align*}
\end{subequations}  
for all $z\in\,\mathcal{Z}$. Here $\alpha_1,\,\alpha_2\,>\,0$, and $\mu_1=1+\frac{1}{\eta}$, $\mu_2=1-\frac{1}{\eta}$ with $\eta\,>\,1$. Moreover, the time of convergence $t_{c}$ adheres $t_{c}\,\leq\,t_{\mathrm{FX}}:=\frac{\eta\,\pi}{2\sqrt{\alpha_1\,\alpha_2}}$, where $t_{\mathrm{FX}}$ is given in~\eqref{rc700}. 
\end{definition}

\smallskip

\begin{remark} \label{rem5}
Since $\mathcal{Z}$ is compact and $h(z)$ is continuously differentiable and positive definite on $\mathcal{Z}$, it yields a uniformly bounded gradient $\sup\limits_{z\in\mathcal{Z}}\big|\big|\nabla\,h(z)\big|\big|\,\leq\,\beta_2$, and a local quadratic growth bound $\beta_3\big|\big|z\big|\big|^2\geq\,h(z)\geq\,\beta_1\big|\big|z\big|\big|^2$ for $\beta_1, \beta_2, \beta_3 > 0$ and for all $z\in\mathcal{Z}\setminus\,\mathcal{S}_G$.
\end{remark}



\smallskip

Following~\cite{artstein1983stabilization} and Definition~\ref{fxt}, we can remark that there exist an admissible, locally lipschitz, continuous-time state-feedback control update $u:\mathcal{Z}\rightarrow\,\mathcal{U}$ for~\eqref{abc}, that can steer the network voltages $\mathbf{v}(t)$ towards the goal set $\mathcal{S}_G$ in fixed time after the initial disturbance occurs. For this controller, we presume the following:

\smallskip

\begin{assumption}\label{bound3}
The control input $u(\cdot)$ in~\eqref{sf100}, satisfies a linear growth bound, i.e., $\big|\big|u(e(t))\big|\big|\leq\xi\big|\big|e(t)\big|\big|$, for $\xi\geq\,0$. 
\end{assumption}

\smallskip


Now, considering this background, we present our first theoretical result, that provides bounds on $\alpha_1,\,\alpha_2$ and $\eta$ in Definition~\ref{fxt}, such that a given $u(t)$ regulates the solution of~\eqref{sf100} towards a sub-level set $\mathcal{S}_G$ with $\mathrm{diam}\left( \mathcal{S}_G\right)<\,1$, at a \textit{fixed time}. 

\begin{theorem} \label{fx100}
Suppose~\eqref{sf100} adheres Assumption~\ref{bound1},~\ref{bound2},~\ref{bound3}, and let $h(e)$ that satisfies Definition~\ref{fxt}, be a FxTS-CLF for the nominal system $\frac{\mathrm{d}}{\mathrm{dt}}e(t)\,=\,\hat{\mathfrak{B}}\,u(t)$. If $\alpha_1,\,\alpha_2,\,\eta$ are selected as $\min\,\left\lbrace \alpha_1,\,\alpha_2 \right\rbrace\,>\,\frac{\xi\,\delta\,\beta_2}{\sqrt{\beta_1}}$ and $\eta\,\in\,\left( 1,\,2\right]$, then there exist an input $u(t)$ and a sub-level set $\mathcal{S}_G$ with $\mathrm{diam}\left( \mathcal{S}_G\right)<\,1$, such that, whenever $\mathbf{v}(t)\notin\mathcal{X}_v$, $u(t)$ converges $\mathbf{v}(t)$ towards the set $\mathcal{S}_G$ in a fixed time $t_{\mathrm{FX}}$. 
\end{theorem}

\begin{proof}
Notably, $d(t)$ in~\eqref{sf100}, is treated as an \textit{impulsive transient}, thus we preclude it in the subsequent analysis for a given fixed-time horizon. Since $h(e)$ is the FxTS-CLF for the \textit{nominal system}, differentiating it with respect to~\eqref{sf100} yields
\begin{align}
\dot{h}(e)\,&=\,\nabla\,h(e)^{\top}\left(\hat{\mathfrak{B}}\,u(t)+\left(\mathfrak{B}-\hat{\mathfrak{B}}\right)\,u(t)+w(t)\right) \nonumber \\
&\leq\,-\alpha_1\,h(e)^{\mu_1}-\alpha_2\,h(e)^{\mu_2}\,+\,\beta_2\,\delta\,\xi\,\big|\big|e(t)\big|\big|\,+\,\underbrace{\beta_2\,\gamma}_{\gamma_1} \nonumber \\
&\leq\,-\alpha_1\,h(e)^{\mu_1}-\alpha_2\,h(e)^{\mu_2}\,+\,\underbrace{\frac{\beta_2\,\delta\,\xi}{\sqrt{\beta_1}}}_{\beta^{\star}}\,h(e)^{\frac{1}{2}}\,+\,\gamma_1. \label{gg}
\end{align}
Here the first inequality is obtained employing Assumption~\ref{bound1},~\ref{bound2} and~\ref{bound3}, while the last inequality is derived based on Remark~\ref{rem5}. Thereafter, we define a set $\mathcal{S}\,:=\,\left\lbrace e(t)\,:\,h(e)\,\leq\,1\right\rbrace$ and adopt two scenarios namely $e(t)\notin\,\mathcal{S}$ and $e(t)\in\,\mathcal{S}$. For the first scenario, rearranging the right hand side of~\eqref{gg}, we obtain 
\begin{multline*}
\dot{h}(e)\,\leq\,-\left( \alpha_1-\beta^{\star}\right)\,h(e)^{\mu_1}\,-\,\alpha_2\,h(e)^{\mu_2}\\
+\,\beta^{\star}\,\underbrace{\left( h(e)^{\frac{1}{2}}-h(e)^{\mu_1}\right)}_{\varepsilon_1(e)}\,+\,\gamma_1.
\end{multline*}
Since, $e(t)\notin\,\mathcal{S}$ (thus, $h(e)>\,1$) and $\eta\in\,\left( 1,\,2\right]$ (thus, $\mu_1\in\,\left[1.5,\,2\right)$), it leads to $\varepsilon_1(e)\leq\,0$. In consequence, the above equation is simplified as
\begin{align}
\dot{h}(e)\,\leq\,-\left( \alpha_1-\beta^{\star}\right)\,h(e)^{\mu_1}-\left( \alpha_2-\beta^{\star}\right)\,h(e)^{\mu_2}\,+\,\gamma_1. \label{c1}
\end{align}
Next considering $e(t)\in\,\mathcal{S}$ and performing algebraic rearrangement on~\eqref{gg} we find 
\begin{multline*}
\dot{h}(e)\,\leq\,-\alpha_1\,h(e)^{\mu_1}\,-\,\left( \alpha_2-\beta^{\star}\right)\,h(e)^{\mu_2}\\
+\,\beta^{\star}\,\underbrace{\left( h(e)^{\frac{1}{2}}-h(e)^{\mu_2}\right)}_{\varepsilon_2(e)}\,+\,\gamma_1.
\end{multline*}
Since $h(e)\leq\,1$ and $\mu_2\in\,\left( 0,\,0.5\right]$ (as $\eta\in\,\left( 1,\,2\right]$), it yields $\varepsilon_2(e)\leq\,0$. Consequently, one gets
\begin{align}
\dot{h}(e)\,\leq\,-\left( \alpha_1-\beta^{\star}\right)\,h(e)^{\mu_1}-\left( \alpha_2-\beta^{\star}\right)\,h(e)^{\mu_2}\,+\,\gamma_1. \label{c2}
\end{align}
By definition, $\min\,\left\lbrace \alpha_1,\,\alpha_2\right\rbrace>\frac{\xi\,\delta\,\beta_2}{\sqrt{\beta_1}}$, thus $\left( \alpha_1-\beta^{\star}\right)>\,0$ and $\left( \alpha_2-\beta^{\star}\right)>\,0$ in~\eqref{c1} and~\eqref{c2}. Moreover, $\mu_1,\,\mu_2$ also satisfy Definition~\ref{fxt}. Then following~\cite[corollary~1]{black2020quadratic}, we can assert that there exists a set $\mathcal{S}_G\subseteq\,\mathbb{R}^n$ such that for all $\mathbf{v}(t_0)\notin\,\mathcal{S}_G$, the trajectories of~\eqref{sf100} reach $\mathcal{S}_G$ in a fixed time $t_{\mathrm{FX}}$. This concludes the proof.
\end{proof} 

\smallskip

In essence, Theorem~\ref{fx100} establishes a dual-stage regulation mechanism: when $e(t)\notin\,\mathcal{S}$, $\mu_1$ dominates $\mu_2$ to drive $e(t)$ toward $\mathcal{S}$; once $e(t)$ enters $\mathcal{S}$, $\mu_2$ takes dominance to guide it into the ultimate goal set $\mathcal{S}_{G}$, see Fig.~\ref{picim}. To characterize this behaviour, the following corollary estimates the boundary of $\mathcal{S}_{G}$ and computes the bounds on the convergence time $t_c$. Hereinafter, for brevity, we define $\alpha_1\,=\,\alpha_2\,=\,\breve{\alpha}$.

\begin{figure}[htb!]
		\centering
		\includegraphics[width=0.8\columnwidth, height=50 mm]{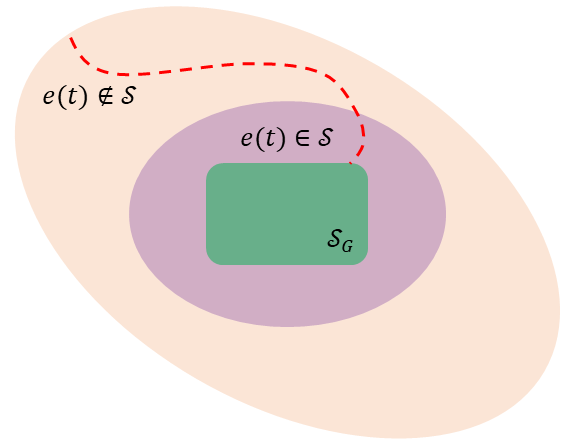}
		\caption{Pictorial implication of Theorem~\ref{fx100} for a given system trajectory.}
	     \label{picim}
		\end{figure}

\begin{corollary} \label{fx200}
Let us define $\mathfrak{h}\,=\,\frac{2\,\left( \breve{\alpha}-\beta^{\star}\right) }{\beta_2\,\gamma}$. Then, whenever $\mathbf{v}(t)\notin\,\mathcal{S}_G$, it converges to $\mathcal{S}_G$ at time $t_c\,\leq\,t_{\mathrm{FX}}$, where $\mathcal{S}_G\,:=\,\left\lbrace e(t)\,:\,h(e)\,\leq\,\frac{1}{\mathfrak{h}}\right\rbrace$ and
\begin{align*}
t_{\mathrm{FX}}\,:=\,\frac{\eta}{\left( \breve{\alpha}-\beta^{\star}\right)\,\sqrt{1-\left( \frac{1}{\mathfrak{h}}\right) ^2}}\,\left[\frac{\pi}{2}\,+\,\tan^{-1}\,\left( \frac{1}{\mathfrak{h}^2-1}\right)^{\frac{1}{2}}\right].  
\end{align*}
\end{corollary}

\smallskip

The proof of this corollary is straightforward and follows~\cite[Theorem~2]{black2020quadratic}. It reveals that $t_{\mathrm{FX}}$ does not depend on the magnitude of the initial disturbance. Instead it relies on the bound of external disturbance $\gamma$, impedance estimate $\delta$ and lipschitz bound $\xi$. 

\smallskip

\begin{remark}
To restrict $\mathrm{diam}\left( \mathcal{S}_G\right)<\,1$, the control gain $\breve{\alpha}$ must be arbitrarily large. Specifically, by Corollary~\ref{fx200}, setting $\breve{\alpha}\,>\,\frac{\beta_2}{2}\left( \gamma\,+\,\frac{2\delta\,\xi}{\sqrt{\beta_1}}\right) $ is sufficient to ensure $\mathrm{diam}\left( \mathcal{S}_G\right)<\,1$. It shows that $\breve{\alpha}$ must dominates both $\gamma$ and $\gamma$ to maintain the desired diameter limit.
\end{remark}

\smallskip

The aforesaid results provide bounds on the control gains and design parameters for guaranteeing fixed-time stability with respect to a goal set $\mathcal{S}_G$, and also estimate it's diameter. Considering this, we attempt to reformulate our problem set-up, such that these results can be leveraged to ensure fixed-time voltage regulation. To this end, the following lemma represents $\mathcal{X}_v$ in terms of $h(e)$

\smallskip

\begin{lemma} \label{map}
Let $h(e)$ be an FxTS-CLF with $\beta_1\,\geq\,1$. Then $e(t)\in\,\mathcal{S}_v$ implies $v(t)\in\,\mathcal{X}_v$, where $\mathcal{S}_v$ is defined as
\begin{align*}
\mathcal{S}_v\,:=\,\left\lbrace e(t)\,:\,h(e)\,\leq\,\left( 0.05\,\mathbf{v}_{\mathrm{nom}}\right)^2\right\rbrace. 
\end{align*}
\end{lemma}

\smallskip

\begin{proof}
When $e(t)\in\mathcal{S}_v$, it holds $h(e)\leq\left( 0.05\,\mathbf{v}_{\mathrm{nom}}\right)^2$. Since, by definition, $\beta_1\geq\,1$, following Remark~\ref{rem5}, we find $\big|\big|e(t)\big|\big|^2\leq\,\left( 0.05\,\mathbf{v}_{\mathrm{nom}}\right)^2$, which further leads to $\big|e(t)\big|\leq\,0.05\,\mathbf{v}_{\mathrm{nom}}$. Simplifying this inequality, one gets $\mathbf{v}(t)\in\,\left[0.95\mathbf{v}_{\mathrm{nom}},\,1.05\mathbf{v}_{\mathrm{nom}}\right]\,=\,\mathcal{X}_v$, which verifies the lemma.
\end{proof}

\smallskip

Lemma~\ref{map} directly implies that if the voltage violation vector $e(t)$ converges to the set $\mathcal{S}_v$, the voltage trajectory $\mathbf{v}(t)$ reaches the safe limit $\mathcal{X}_{v}$. Leveraging this relationship, Theorem~\ref{fx400} derives a bound on $\breve{\alpha}$ to ensure $\mathcal{S}_G \subseteq \mathcal{S}_v$, thereby guaranteeing fixed-time voltage stability. Throughout this analysis, we assume $\mathbf{v}_{\mathrm{nom}}\,=\,1\,\mathrm{p.u}$ without loss of generality.

\smallskip

\begin{theorem} \label{fx400}
If $\breve{\alpha}\geq\,\beta_2\,\left( 10\,+\,\frac{\delta\,\xi}{\sqrt{\beta_1}}\right)$, then $\mathcal{S}_G\subseteq\,\mathcal{S}_v$. 
\end{theorem}

\smallskip

\begin{proof}
In order to establish $\mathcal{S}_G\subseteq\,\mathcal{S}_v$, following Theorem~\ref{fx200} and~Lemma~\ref{map}, we need to ensure  $\frac{1}{\mathfrak{h}}\,\leq\,\left( 0.05\,\mathbf{v}_{\mathrm{nom}}\right)^2$. Revisiting $\mathfrak{h}$ from Corollary~\ref{fx200}, and defining $\gamma\,=\,0.05\,\mathbf{v}_{\mathrm{nom}}$, the prior inequality can be written as $\left( \breve{\alpha}-\frac{\beta_2\,\delta\,\xi}{\sqrt{\beta_1}}\right)\geq\,10\,\beta_2$. Thereafter an algebraic rearrangement completes the proof.  
\end{proof}

\smallskip

Theorem~\ref{fx100} computes $u(t)$, that steers $\mathbf{v}(t)$ to the goal set $\mathcal{S}_{G}$ within a fixed time $t_{\mathrm{FX}}$, while Theorem~\ref{fx400} bounds $\breve{\alpha}$ to ensure $\mathcal{S}_G \subseteq \mathcal{S}_v$. As Lemma~\ref{map} establishes $\mathcal{S}_v \subseteq \mathcal{X}_v$, convergence to $\mathcal{S}_{v}$ implies voltage stability. Finally, Corollary~\ref{fx200} guarantees fixed-time voltage regulation by proving the convergence of $e(t)$ towards $\mathcal{S}_{G}\subseteq \mathcal{S}_v\subseteq \mathcal{X}_v$ within fixed-time $t_{\mathrm{FX}}$.


\subsection{Limitation of Fixed-Time Convergence under Bounded Control}
The analysis presented in Theorem~\ref{fx100} and~\ref{fx400} do not consider any bounds on the control input to ensure a \textit{uniformly bounded voltage recovery time}, whereas, in practice, distribution network controllers often exhibit limited control capabilities due to strict apparent power and thermal limits. Large, rapid grid disturbances can easily saturate controller outputs, nullifying the fixed-time convergence and transforming the closed-loop system to a \textit{finite-time} paradigm until the voltage violation $e(t)$ returns to the \textit{safe operating envelope}. Theorem~\ref{fts} explicitly formalizes this behaviour; ahead of this, we define the \textit{safe operating envelope} for completeness.

\smallskip

\begin{definition}
The safe operating region $\Omega _{0}$, is defined as the following compact set
\begin{align*}
\Omega_0 \triangleq \left\{ e(t) \in \mathbb{R}^n \,:\,~||e(t)||\,\leq\,e_{\max}\right\rbrace,
\end{align*}
where $e_{\max}$ is the maximum allowable euclidean norm of $e(t)$ for which the fixed-time control law strictly stays within the control limit, i.e., $u(t)\in\mathcal{X}_u\,:=\,\left[-u_{\max},\,u_{\max}\right]$, guaranteeing fixed-time convergence to $\mathcal{S}_G$.
\end{definition}

\smallskip

\begin{theorem} \label{fts}
Suppose $e(0)\notin\,\Omega_0$ due to an arbitrarily large grid disturbance $d(t)$ in~\eqref{sf100}, where $d(t)$ satisfies Assumption~\ref{mda}. We define a (saturated) control update $u_{\mathrm{sat}}(t)\in\mathcal{X}_u$, as
\begin{equation}
u_{\mathrm{sat}}(t) = -u_{\max} \frac{\mathfrak{B}^{-1}e(t)}{||\mathfrak{B}^{-1}e(t)||}, \label{inp}
\end{equation}
where $u_{\max}\,=\,\frac{1}{2}\mathrm{diam}\left( \mathcal{X}_u\right)$. If $u_{\mathrm{sat}}(t)$ in~\eqref{inp}, dominates $d_{\max}$ such that $\sigma_{\min}(\mathfrak{B})\,u_{\max}\,>\,d_{\max}$, where $\sigma_{\min}(\mathfrak{B})$ refers the minimum singular value of $\mathfrak{B}$, then, for any $e(0)\notin\,\Omega_0$, it allows $e(t)$ to converge $\Omega_0$ within a finite settling time $t_{\text{FT}}$, which is upper-bounded by
\begin{align*}
t_{\mathrm{FT}}\,\leq\,\frac{||e(0)||\,-\,e_{\max}}{\sigma_{\min}(\mathfrak{B})\,u_{\max}\,-\,d_{\max}}.
\end{align*}
\end{theorem}

\smallskip

\begin{proof}
Since $e(0)\notin\,\Omega_0$, it saturates the control input $u(t)$. Thus, substituting $u_{\mathrm{sat}}(t)$ in~\eqref{inp} to~\eqref{sf100}, we find
\begin{align}
\dot{e}(t)\,=\,-u_{\max} \frac{e(t)}{||\mathfrak{B}^{-1}e(t)||}\,+\,w(t)\,+\,d(t). \label{ad}
\end{align}
Thereafter, selecting the Lyapunov function $V(e)\,=\,\frac{1}{2}e^{\top}\,e$, and computing its derivative with respect to~\eqref{ad} leads to
\begin{align}
\dot{V}(e(t))\,&=\,-u_{\max} \frac{||e(t)||^2}{||\mathfrak{B}^{-1}e(t)||}\,+\,e(t)^{\top}\left( w(t)\,+\,d(t)\right)\nonumber \\
& \leq\,- ||e(t)||\,\underbrace{\left( \sigma_{\min}(\mathfrak{B})\,u_{\max}\,-\,d_{\max}\right)}_{\theta}. \label{abcd}
\end{align}
The above inequality is obtained employing \textit{singular value} and \textit{Cauchy-Schwarz inequality} and presuming $||d(t)||\gg ||w(t)||$, which further reduces to $\dot{V}(e(t))\,\leq\,-\sqrt{2}\,\theta\,V(e(t))^{\frac{1}{2}}$. Since, by definition, $\theta\,>\,0$ in~\eqref{abcd}, it infers that whenever $e(0)\notin\,\Omega_0$, the input $u_{\mathrm{sat}}(t)$ brings back $e(t)$ inside $\Omega_0$ in \textit{finite time}. Simplifying $\dot{V}(e(t)) = ||e(t)||\frac{d||e(t)||}{dt}$, and dividing both sides of~\eqref{abcd} by $||e(t)||$, we find
\begin{align*}
\frac{d||e(t)||}{dt}\,\leq\,-\left( \sigma_{\min}(\mathfrak{B})\,u_{\max}\,-\,d_{\max}\right).
\end{align*}
Integrating both sides from $t = 0$ to $t = t_{\mathrm{FT}}$ and across the spatial transition from $||e(0)||$ down to the targeted linear boundary radius $e_{\max}$ provides
\begin{align*}
e_{\max}\,-\,||e(0)||\,\leq\,-\left( \sigma_{\min}(\mathfrak{B})\,u_{\max}\,-\,d_{\max}\right)\,t_{\mathrm{FT}}.
\end{align*}
Rearranging the terms to isolate $t_{\mathrm{FT}}$ yields the upper bound. This concludes the proof.
\end{proof}

\smallskip

Theorem~\ref{fts} shows that if the voltage violation vector initializes outside the safe operating limit, i.e., $e(0)\,\notin\,\Omega_0$, the controller outputs will get saturated. Consequently, the fixed-time stability structure temporarily collapses, transforming the closed-loop system into a pure finite-time reaching law. Then the system pays a linear temporal penalty $t_{\mathrm{FT}}$ that scales directly with the initial violation vector $||e(0)||$ until the trajectory enters $\Omega_0$, at which point the unconstrained fixed-time convergence mechanism is restored. This behaviour has also been depicted in Fig.~\ref{picim5} for completeness.

\begin{figure}[htb!]
		\centering
		\includegraphics[width=1\columnwidth, height=50 mm]{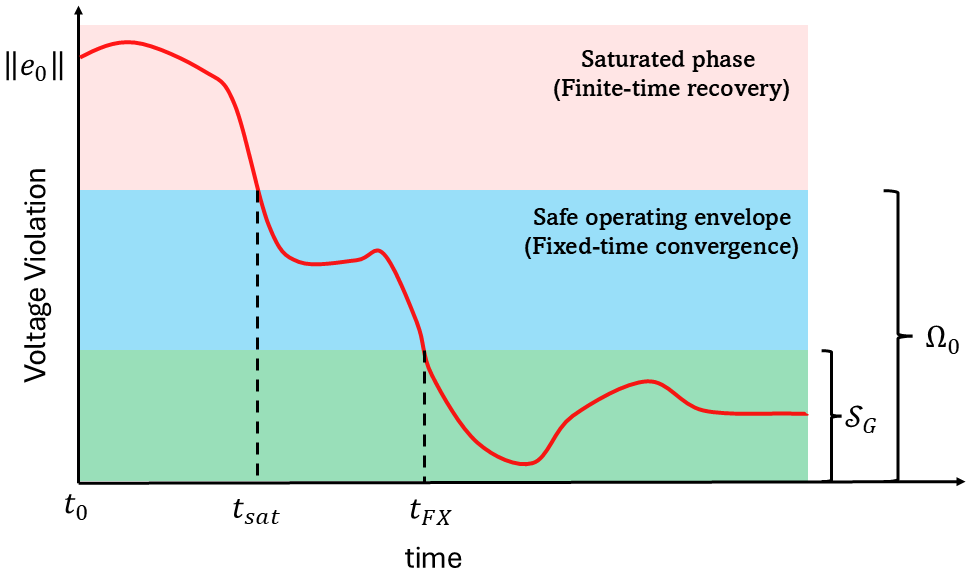}
		\caption{Illustration of Theorem~\ref{fts} for a given voltage violation trajectory.}
	     \label{picim5}
		\end{figure}
		
%
%

\vspace{-0.4 cm}

\subsection{Synthesizing Fixed-Time Control under Input Constraints}
This section synthesizes $u(t)$ under input constraint, that solves~\eqref{rc100} in fixed time. The proposed approach utilizes the interplay between~FxTS-CLF in Definition~\ref{fxt}, and QP. It first transforms~\eqref{rc100} to an equivalent QP problem and then translates the conditions in Theorem~\ref{fx100} and~\ref{fx400} to the constraint set. Considering these modifications, for $\mathbf{v}(t)\notin\,\mathcal{X}_v$, we equivalently formulate~\eqref{rc100} as the following QP problem:
\begin{subequations} \label{qp}
\begin{align}
\min_{u,\,\xi} \quad \frac{1}{2}u^\top{P}u\,& +\,\left( \frac{1}{2}+c\right) \nu^2 \label{qp1}\\
\textrm{s.\,t.} \quad \begin{bmatrix}
\mathbf{I}_{\mathcal{C}} \\ -\mathbf{I}_{\mathcal{C}}
\end{bmatrix}\,u\,& \leq\,\begin{bmatrix}
u_{\max} \\ -u_{\max}
\end{bmatrix} \label{qp2}\\
\nabla\,h(e)^\top\mathrm{\hat{\mathfrak{B}}}u\,\leq\,&\nu\,h(e)\,-\alpha_1\,\max\left\lbrace 0,\,h(e)\right\rbrace^{\mu_1} \nonumber \\
 &-\alpha_2\,\max\left\lbrace 0,\,h(e)\right\rbrace^{\mu_2}.\label{qp3}
\end{align}
\end{subequations}
Here $c>\,0,\,P\succ\,0$, and $\nu\in\mathbb{R}$ is the slack variable. Per Theorems~\ref{fx100} and \ref{fx400}, we select $\mu_1,\,\mu_2$, and set $\min\left\lbrace \alpha_1,\,\alpha_2\right\rbrace\geq\,\beta_2\,\left( 10\,+\,\frac{\delta\,\xi}{\sqrt{\beta_1}}\right)$ to ensure fixed-time voltage stability. Crucially, these choices do not affect feasibility of~\eqref{qp}. The set $\mathcal{X}_{u}$ in~\eqref{rc500}, bounded by $\left[ -u_{\max},\,u_{\max}\right] $ in~\eqref{qp2}, which restricts inputs throughout the optimization window. Solving~\eqref{qp} pointwise yields $u(t)$, which regulates $\mathbf{v}(t)$ toward $\mathcal{X}_{v}$ within a fixed time. Notably, if one selects $\nu\,=\,0$ in~\eqref{qp}, the input constraints~\eqref{qp2} can make~\eqref{qp} infeasible for some $\mathbf{v}(t)$. Therefore, the slack variable $\nu$ is introduced to uphold the feasibility of~\eqref{qp}; see Theorem~\ref{fa} for the feasibility study. In essence, the above formulation works as a fixed time control synthesizer in scenarios where the control inputs are large enough.while for other scenarios, it automatically synthesizes finite time control actions.

\smallskip

\begin{theorem} \label{fa}
Problem~\eqref{qp} is feasible for all $\mathbf{v}(t)\notin\,\mathcal{X}_v$, i.e., for each $\mathbf{v}(t)\notin\,\mathcal{X}_v$, there exist a $\left( u(t),\,\nu\right)\in\,\mathbb{R}^{n+1} $, which satisfies~\eqref{qp2} and~\eqref{qp3}. 
\end{theorem}

\smallskip

\begin{proof}
Select any $\bar{u}\in\,\left[ -u_{\max},\,u_{\max}\right] $, such that it satisfies the input constraint in~\eqref{qp2}. Furthermore, since $h(e)\,>\,0$ for $\mathbf{v}(t)\notin\,\mathcal{X}_v$, we find that
\begin{align*}
\bar{\nu}(e)\,=\,\frac{\nabla\,h(e)^\top\mathrm{\hat{\mathfrak{B}}}\bar{u}+\alpha_1\,h(e)^{\mu_1}+\alpha_2\,h(e)^{\mu_2}}{h(e)}
\end{align*}
is well-defined. Thus, for $u(t)\,=\,\bar{u}$ and $\nu\,=\,\bar{\nu}$, constraint~\eqref{qp3} also holds with equality. Therefore, $\left( \bar{u},\,\bar{\nu}\right) $ satisfies~\eqref{qp2} and~\eqref{qp3} for all $\mathbf{v}(t)\notin\,\mathcal{X}_v$, which verifies the feasibility of~\eqref{qp}.
\end{proof}

\begin{remark}
Our underlying control framework can easily be integrated with any offline or online estimation technique that yields a bounded parameter estimation error $\left( \mathfrak{B}-\hat{\mathfrak{B}}\right)$, thereby satisfying Assumption~\ref{bound2}. For instance, one could deploy the classical offline iterative least-squares algorithm~\cite{picallo2022adaptive}, though its performance remains sensitive to data quality and quantity. Conversely, the online convex-body chasing method in~\cite{yeh2024online,chowdhury2025provably} overcomes this data dependency, offering remarkable robustness against modelling errors and adversarial noise despite lower optimal performance in large-scale systems.
\end{remark}

\section{Numerical Case Studies} \label{ncs}
This section tests our algorithms on a benchmark IEEE distribution network~\cite{dolatabadi2020enhanced} to verify its efficacy.

\noindent \textbf{Network setup:}~We consider the standard IEEE $33$ bus distribution network, which has $\mathbf{v}_{\text{nom}}\,=\,12.66~\text{kV}$, and $\left[ \mathbf{v}_{\min},\,\mathbf{v}_{\max}\right]=\left[ 0.95,\,1.05\right]~\text{pu}=\left[ 12.03,\,13.29\right]~\text{kV}$. To induce voltage fluctuation in this network, we introduce  \textit{four} solar PV panels of capacity $4$ MW (each) at bus $2,\,5,\,12$ and $18$, see Fig.~\ref{33bds}. We simulate daytime \textit{over-voltage} scenarios driven by peak PV generation and night time \textit{under-voltage} conditions caused by high load demands without adequate generation. The simulation spans a consecutive $48$-hour window. The solar injection data (sourced from~\cite{data150719}) and load profiles are sampled at \(15\)-minute intervals and shown in Fig.~\ref{fig1000} (Left plot).

\begin{figure}[htb!]
		\centering
		\includegraphics[width=0.8\columnwidth,trim=0cm 0cm 0.0cm 0cm, clip]{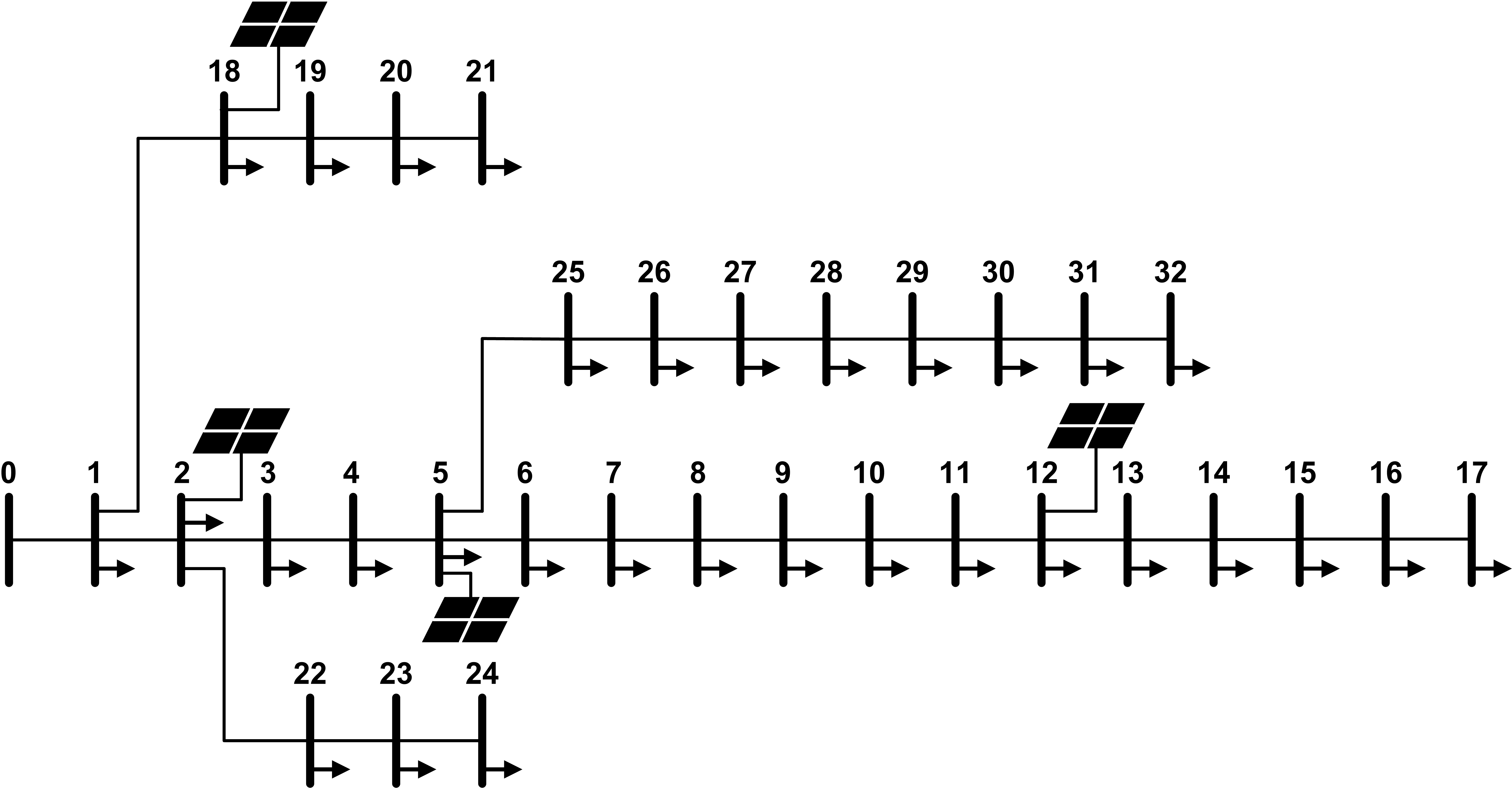}
		\caption{\small{Illustration of the IEEE $33$ bus distribution network.}} \label{33bds}
		\end{figure}	

We compute $\mathbf{v}(t)$ of some particular buses in Fig~\ref{fig1000} (Right plot), using \textit{pandapower}~\cite{thurner2018pandapower}, which solves the exact nonlinear branch flow equations of the (modified) network. We observe that integration of PV panels to the network deviates voltage beyond safe limits, causing both over- and under-voltage violations, which motivates us to deploy \eqref{qp}.

\begin{figure}[htb!]
		\centering
		\includegraphics[width=1.0\columnwidth, height=40 mm]{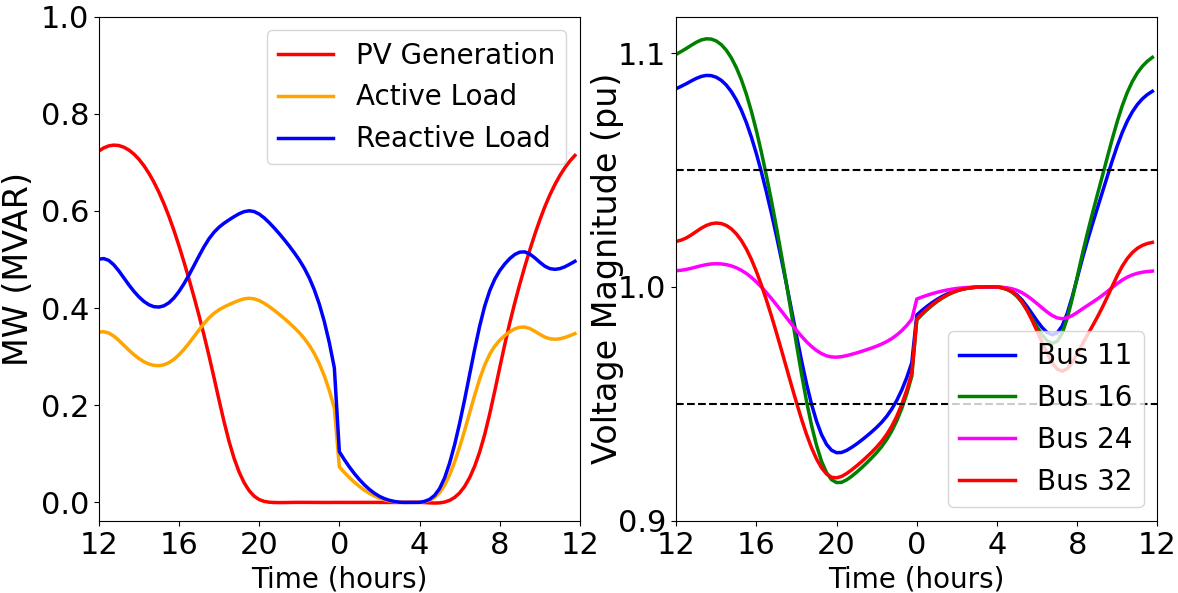}
		\caption{\small{Left:~Active, reactive power demand profile, solar PV profile of the underlying distribution network. Right: uncontrolled voltage profile of bus $11, 16, 24$ and $32$.}} \label{fig1000}
		\end{figure}
		
\noindent \textbf{Choice of appropriate baselines:}~To verify the efficacy of~\eqref{qp}, we consider two recent \textit{baselines}, e.g., i) Robust Predictive Control (\textit{Robust PC})~\cite{yeh2024online}:~an online framework that jointly learns network model and stable controller, and attains finite-mistake guarantee for real-time voltage violation. ii) Exponential CLF (\textit{E-CLF})-based control~\cite{colot2024optimal}:~An algorithm that regulates network voltage using control barrier functions and projected gradient flows under a linearly approximated bus injection model.

\noindent \textbf{Regulation with time-varying load and PV generation:}~We adopt the demand and solar profile in Fig.~\ref{fig1000}, and solve~\eqref{qp} considering the exact values of the $\mathfrak{B}$ matrix (thus, $\delta\,=\,0$ in Assumption~\ref{bound2}). Here we define the CLF as $h(e)=\sum_{i=1}^n\,\frac{1}{2}e_i^2\,-\,0.0625\,\left(v^i_{\max}-v^i_{\min}\right)^2$, and choose $P_u\,=\,\mathbf{I}_n$, $c=\,10$. Following Theorem~\ref{fx100},~\ref{fx400}, we set $\alpha_1=\,\alpha_2=\,10$, $\mu_1=\,1.9$ and $\mu_2=\,0.1$. Moreover, we randomly select $\mathcal{Q}\,=\,\left[-q~\text{MVar},\,q~\text{MVar}\right]$, where $q\in\,\left\lbrace 0.05,\cdots, 0.1\right\rbrace$. 

Following this, in Fig.~\ref{fig2}, we compute $q(t)$ and $\mathbf{v}(t)$ trajectories for bus $16$, which shows that while all methods respect reactive power limits, only our approach stabilizes network voltages. This is because baseline algorithms ignore network losses, whereas our method accounts for them approximately to design the regulation algorithm. We further observe spikes in the reactive power trajectories, which are likely driven by power flow non-linearities, input constraints, and external disturbances. Additionally, Table~\ref{vrt} computes the \textit{voltage recovery times} (in seconds) for bus $16$, which substantiates that across all the design parameter values, the FxTS-CLF algorithm yields a significantly faster recovery time than the other methods. In addition, it is also observed that, in practice, the network  voltage converges to the safe limit much faster considering the theoretical recovery time obtained via Corollary~\ref{fx200}.

\begin{figure}[htb!]
		\centering
		\includegraphics[width=1.0\columnwidth, height=40 mm]{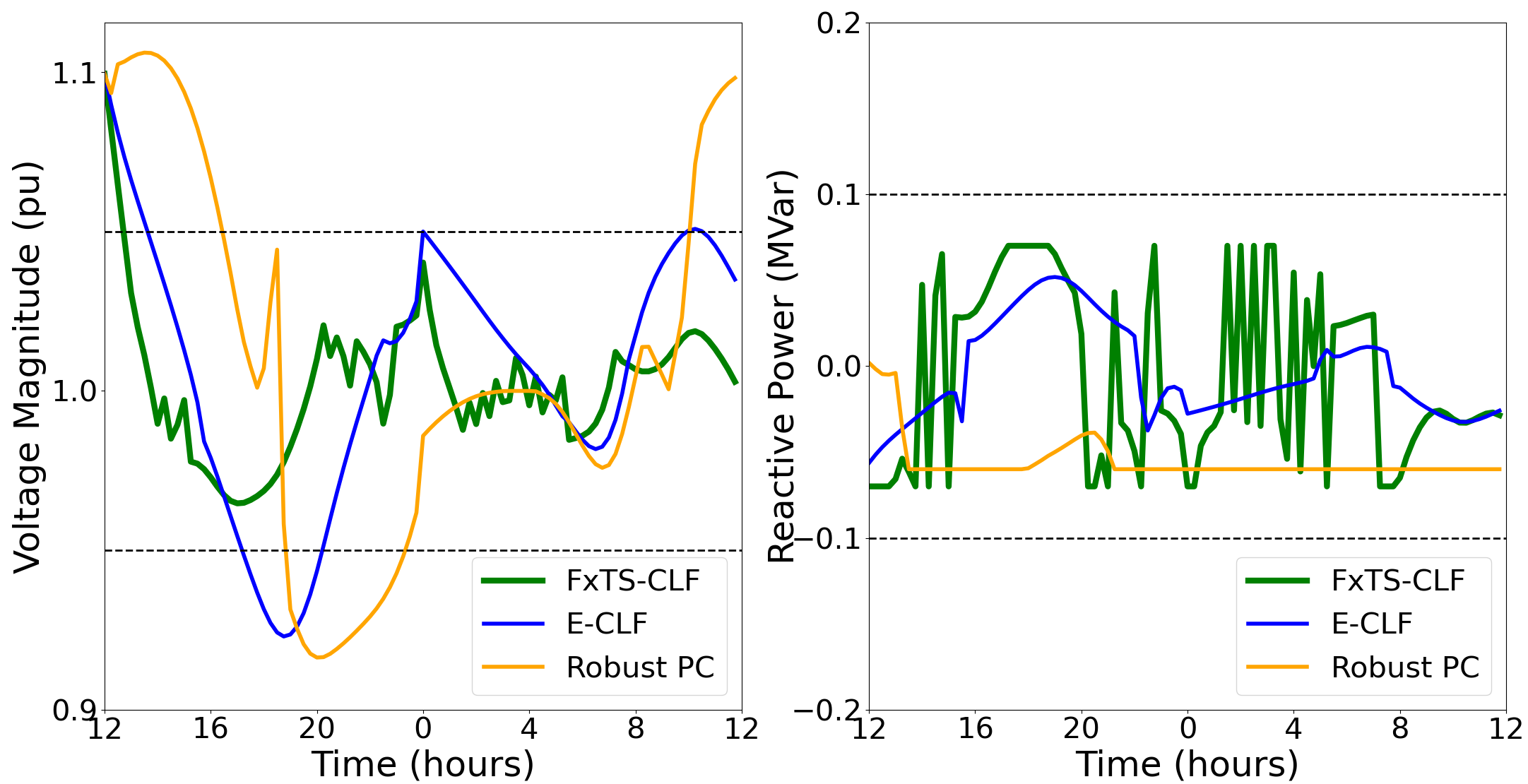}
	     \caption{\small{Comparison with baselines for dynamic load profile: Left and right plots display the voltage and reactive power trajectories of bus 16, respectively.}} \label{fig2}
		\end{figure}

\begin{center}
	\begin{table} [htb!]
		\caption{Voltage recovery time (in seconds) of bus $16$} \label{vrt}
		\centering
		\begin{tabular}{ | p{1.0 cm} | p{2.0 cm} | p{2.0 cm} | p{2.0 cm} | } 
			\hline
			Method & Control Coefficient & Practical recovery time (sec) & Theoretical Recovery time (sec)\\
			\hline
			\multirow{5}{4em}{FxTS-CLF} & $\left( \mu_1,\,\mu_2\right)\,=\,\left( 1.9,\,0.1\right)  $ & $7.27$ & $8.93$ \\ 
			& $\left( \mu_1,\,\mu_2\right)\,=\,\left( 1.7,\,0.1\right)  $ & $6.89$ & $8.25$ \\ 
			& $\left( \mu_1,\,\mu_2\right)\,=\,\left( 1.5,\,0.1\right)  $ & $6.39$ & $7.67$\\
			& $\left( \mu_1,\,\mu_2\right)\,=\,\left( 1.9,\,0.3\right)  $ & $6.7$ & $8.14$ \\
			& $\left( \mu_1,\,\mu_2\right)\,=\,\left( 1.9,\,0.5\right)  $ & $6.5$ & $7.92$\\  \hline
			E-CLF & $\mu_1\,=\,1$ & $17.84$ & $--$\\
			\hline
			Robust-PC & $--$ & $36.582$ & $--$\\
			\hline\hline
		\end{tabular}
		\label{tab1}
	\end{table} 
\end{center}	
		
\noindent \textbf{Robustness verification:}~Next, we evaluate the robustness of our method against network impedance uncertainty (by doubling the feeder resistance and reactance) and measurement errors (corrupting the real-time voltage measurements by a zero-mean random Gaussian noise with $10\%$ standard deviation, i.e., $1.266$ kv). To account for impedance discrepancies, we integrate the estimation technique from~\cite{yeh2024online}, which guarantees bounded estimation errors and maintain the same simulation data and baseline algorithms for a fair comparison. We found $\delta\,=\,8.27$ for~\cite{yeh2024online}, and set $\beta_1\,=\,\beta_2\,=\,\xi\,=\,1$. Then following Theorem~\eqref{fx400}, we select $\alpha_1\,=\,\alpha_2\,=\,20$. The illustrations are depicted in Fig.~\ref{fig4}, and it shows that our proposed algorithm successfully regulates network voltages within safe limits despite these uncertainties, whereas the baseline algorithms fail to do so.

\begin{figure}[htb!]
  \centering
	\includegraphics[width=1.0\columnwidth, height=40 mm]{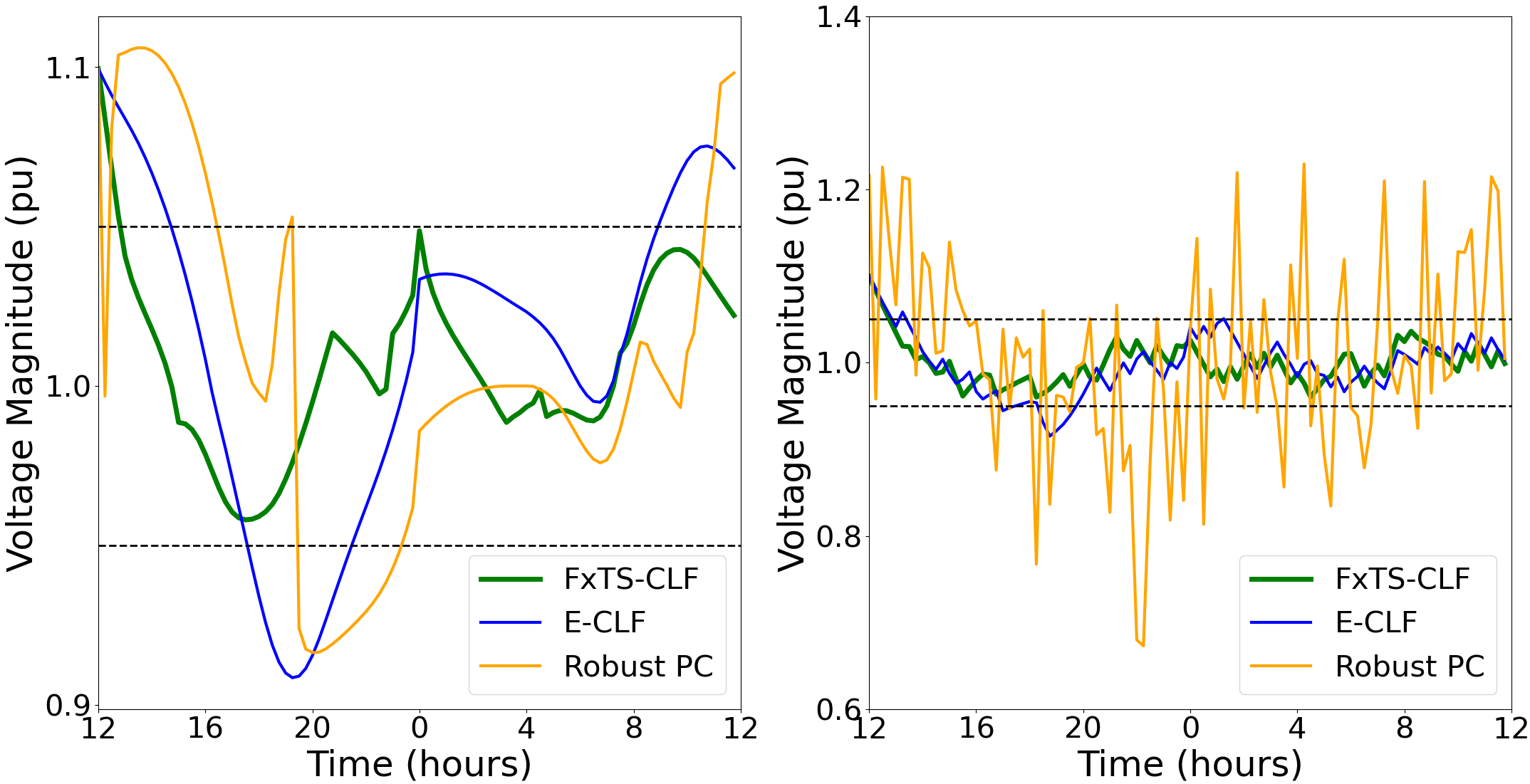}
	  \caption{\small{Robustness verification of $\mathrm{FxTS-CLF}$ algorithm: left and right plot, respectively elucidate robustness against impedance uncertainty and measurement noise for bus $16$.}} \label{fig4}
\end{figure}


\section{Concluding Remarks}
This article introduces an impedance-cognizant, fixed-time control algorithm for solving real-time voltage regulation problem of a radial, balanced distribution networks. We first leverage fixed-time stability (FxTS) and control Lyapunov function (CLF) techniques, to derive analytical bounds on control gains and design parameters to guarantee robust stability under external disturbances. Thereafter, translate the aforesaid conditions into an efficient QP framework to find control set points. We also demonstrate that when initial voltage violations exceed safe boundaries, controller saturation temporarily degrades the fixed-time structure into a finite-time reaching law, incurring a linear time penalty proportional to the violation magnitude before unconstrained FxTS convergence recovers. Empirical validation on an IEEE-33 bus network confirms the algorithm's capability to maintain voltage stability against measurement and modelling uncertainties.

\small
\bibliographystyle{IEEEtran}
\bibliography{publications,refs} 
\end{document}